\documentclass[conference]{IEEEtran}
\usepackage{amsmath,amssymb,amsthm,cite,graphicx,array}
\usepackage{graphicx}
\usepackage{float}
\usepackage{caption}
\usepackage{multicol}
\usepackage{mathrsfs}
\usepackage{amsmath}
\usepackage{amssymb}
\usepackage{amsthm}
\usepackage{multicol}
\usepackage{color}
\usepackage{float}
\usepackage{placeins}
\usepackage{epstopdf}
\usepackage{multirow}
\usepackage{subfigure}
\usepackage{enumitem}
\usepackage[utf8]{inputenc}
\floatstyle{ruled}
\theoremstyle{remark}
\newtheorem{theorem}{Theorem}

\newtheorem{lemma}{Lemma}

\newtheorem{condition}{Condition}

\theoremstyle{remark}
\newtheorem{example}{Example}
\ifCLASSINFOpdf
\else
\fi

\title{A Generalisation of Interlinked Cycle Structures and Their Index Coding Capacity}

\begin{document}

\author{Mahesh~Babu~Vaddi~and~B.~Sundar~Rajan\\ 
 Department of Electrical Communication Engineering, Indian Institute of Science, Bengaluru 560012, KA, India \\ E-mail:~\{vaddi,~bsrajan\}@iisc.ac.in }
 
\maketitle
\begin{abstract}
Cycles and Cliques in a side-information graph reduce the number of transmissions required in an index coding problem. Thapa, Ong and Johnson defined a more general form of overlapping cycles, called the interlinked-cycle (IC) structure, that generalizes cycles and cliques. They proposed a scheme, that leverages IC structures in digraphs to construct scalar linear index codes. In this paper, we extend the notion of interlinked cycle structure to define more generalised graph structures called overlapping interlinked cycle (OIC) structures. We prove the capacity of OIC structures by giving an index code with length equal to the order of maximum acyclic induced subgraph (MAIS) of OIC structures.   
\end{abstract}

\section{Introduction and Background}
\label{sec1}
A single unicast index coding problem, comprises a transmitter that has a set of $K$ messages, $X=\{x_1,x_2,\ldots,x_K\}$, and a set of $K$ receivers, $R=\{R_1,R_2,\ldots,R_K\}$. Each receiver, $R_k=(\mathcal{W}_k,\mathcal{K}_k)$, knows a subset of messages, $\mathcal{K}_k \subseteq X$, called its side-information, and wants to know one message, $W_k =\{x_k\}$, called its \textit{Want-set}. The transmitter can take cognizance of the side-information of the receivers and broadcast coded messages, called the index code. The objective is to minimize the number of coded transmissions, called the length of the index code, such that each receiver can decode its demanded message using its side-information and the coded messages.  

The index coding with side-information was introduced by Birk and Kol in \cite{BiK}. Single unicast index coding problems were studied in \cite{YBJK}. A single unicast index coding problem (SUICP) can be represented by using a graph $G$ with $K$ vertices $\{x_1,x_2,\ldots,x_K\}$. In $G$, there exists an edge from $x_i$ to $x_j$ if the receiver wanting $x_i$ knows $x_j$. This graph is called the side-information graph of SUICP.

In an index coding problem with side-information graph $G$, we assume that the messages belongs to a finite alphabet $\mathcal{A}$. The solution of an index coding problem may be linear or nonlinear. The solution must specify a finite alphabet $\mathcal{A}_P$ to be used by the transmitter, and an encoding scheme $\varepsilon:\mathcal{A}^{K} \rightarrow \mathcal{A}_{P}$ such that every receiver is able to decode the wanted message from the $\varepsilon(x_1,x_2,\ldots,x_K)$ and the side-information. The minimum encoding length $l=\lceil log_{2}|\mathcal{A}_{P}|\rceil$ for messages that are $t$ bit long ($\vert\mathcal{A}\vert=2^t$) is denoted by $\beta_{t}(G)$. The broadcast rate of the index coding problem is defined \cite{ICVLP} as,
\begin{align*}
\beta(G) \triangleq   \inf_{t} \frac{\beta_{t}(G)}{t}.
\end{align*}

The capacity $C(G)$ for the index coding problem with side-information graph $G$ is defined as the maximum number of message symbols transmitted per index code symbol such that every receiver gets its wanted message symbol. The broadcast rate and capacity are related as 
\begin{align*}
C(G)=\dfrac{1}{\beta(G)}.
\end{align*}

In this paper, we refer the capacity of a single unicast index coding problem with side-information graph $G$ as the index coding capacity of side-information graph $G$. In \cite{MCJ}, Maleki \textit{et.al.} found the index coding capacity of some side-information graphs which have a circular symmetry by using interference alignment technique. However, in general, finding the index coding capacity is a complicated problem because one need to consider all possible linear and non linear mappings and dimensions to evaluate capacity. 

For a graph $G$, the order of an induced acyclic sub-graph formed by removing the minimum number of vertices in $G$, is called Maximum Acyclic Induced Subgraph ($MAIS(G)$). In \cite{YBJK}, it was shown that $MAIS(G)$ lower bounds the broadcast rate of the index coding problem described by $G$. That is, 
\begin{align}
\label{mr1}
\beta(G) \geq MAIS(G).
\end{align}

\subsection{Interlinked Cycles and Optimal Index Codes} 
\label{sec5}
In \cite{ICC}, Thapa, Ong and Johnson defined a special graph structures called interlinked cycle (IC) structure. The interlinked cycle structures generalises the notion of cycles and cliques. Consider a graph $G$ with $K$ vertices $\{x_1,x_2,\ldots,x_K\}$ having the following property: $G$ has a vertex set $V_I$ such that for any ordered pair $(x_i \in V_I,x_j \in V_I)$ and $x_i \neq x_j$, there is a path from $x_i$ to $x_j$, and the path does not include any other vertex in $V_I$ except $x_i$ and $x_j$. The set $V_I$ is called inner vertex set and the vertices in $V_I$ are called inner vertices. A path in which only the first and the last vertices are from $V_I$, and they are distinct, is called an $I$-path. If the first and last vertices are the same, then it is called an $I$-cycle. If the directed graph $G$ satisfies the four conditions given below, it is called an interlinked-cycle structure.

\begin{itemize}
\item There is no $I$-cycle in $G$.
\item Every non-inner vertex must be present in atleast one $I$-path.
\item For all ordered pairs of inner vertices $(x_i,x_j)$, $x_i \neq x_j$, there is only one $I$-path from $x_i$ to $x_j$ in $G$.
\item There exist no cycles among non-inner vertices
\end{itemize}

Let $G$ be the IC structure with $K$ vertices $\{x_1,x_2,\ldots,x_{K}\}$ and $N$ inner vertices $V_I=\{x_1,x_2,\ldots,x_{N}\}$. Let the $K-N$ non-inner vertices be $V_{NI}=\{x_{N+1},x_{N+2},\ldots,x_K\}$. The following coded symbols for an IC structure $G$ with $\vert V(G) \vert=K$ was proposed in \cite{ICC}.

\begin{itemize}
\item A code symbol is obtained by the bitwise XOR (denoted by $\oplus$) of messages present in the inner vertex set $V_I$, i.e.,
\begin{align}
\label{code1}
y_I=\bigoplus_{i=1}^{N} x_i.
\end{align}
\item For each $x_j \in V_{NI}$, for $j \in [N+1:K]$, a code symbol is obtained as given below.
\begin{align}
\label{code2}
y_{j}=x_j \bigoplus_{x_q \in N_{G}^+(x_j)} x_q.
\end{align}
where $N_{G}^+(x_j)$ is the out-neighborhood of $x_j$ in the IC structure $G$.
\end{itemize}

The length of index code constructed above is $K-N+1$. That is, for an ICP whose side-information graph is an IC structure $G$ with $K$ vertices, the index code given in \eqref{code1} and \eqref{code2} give a savings of $N-1$ transmissions when compared with naive technique of broadcasting all $K$ messages. Thapa, Ong and Johnson proved that the constructed codes in \eqref{code1} amd \eqref{code2} are of optimal length.

The following decoding procedure is given in \cite{ICC} to decode the index codes given by ICC scheme.
\begin{itemize}
\item For $j \in [N+1:K]$, the message $x_j$ ($x_j$ corresponding to a non-inner vertex) can be decoded from $y_j$ given in \eqref{code2}.
\item For each $x_k \in V_I$, a directed rooted tree (denoted by $T_k$) in $G$ can be found with $x_k$ as the root vertex and all other inner vertices $V_I \setminus \{x_k\}$ are the leaves. The inner vertex $x_k \in V_I$ is decoded by computing the XOR of all index code symbols corresponding to the non-leaf vertices at depth greater than zero in $T_k$ and $y_I$, where $T_k$ is the rooted tree with $x_k$ as the root node and all other inner vertices as the leaves.
\end{itemize}
The following two examples illustrates IC structures and their decoding.
 
\begin{example}
\label{ex1}
Consider an SUICP with side-information graph $G$ given in Fig. \ref{ex1fig}. $G$ is an IC structure with $K=5,N=3$ and inner vertices $V_I=\{x_1,x_2,x_3\}$. Hence, for this side-information graph, we have $K-N+1=3$. An optimal length index code for this ICP obtained from \eqref{code1} and \eqref{code2} is $$\mathfrak{C}=\{\underbrace{\underbrace{x_1+x_2+x_3}_{\in V_I}}_{y_I},~\underbrace{\underbrace{x_4}_{\in V_{NI}}+\underbrace{x_2}_{N_{G}^+(x_4)}}_{y_4},~\underbrace{\underbrace{x_5}_{\in V_{NI}}+\underbrace{x_1}_{N_{G}^+(x_5)}\}}_{y_5}.$$ 
Trees $T_1,T_2$ and $T_3$ corresponding to the inner vertices $x_1,x_2$ and $x_3$ are given in Figure \ref{ex1fig2}. The decoding of each message symbol from $\mathfrak{C}$ is summarised in Table \ref{table111}. 

In the rest of the paper, we use $\gamma_k$ to denote the index code symbols used by $R_k$ to decode $x_k$ and $\tau_k$ to denote the sum of index code symbols present in $\gamma_k$. In the IC structures, we use orange color to identify inner vertices.

\begin{table}[ht]
\centering
\setlength\extrarowheight{3pt}
\begin{tabular}{|c|c|c|c|}
\hline
$x_k$ & Tree &$\gamma_k$ & $\tau_k$ \\
\hline
$x_1$ & $T_1$& $y_I,y_4$& $x_1+\underbrace{x_3+x_4}_{\text{side-information}}$ \\
\hline
$x_2$ & $T_2$ & $y_I,y_5$ & $x_2+\underbrace{x_3+x_5}_{\text{side-information}}$  \\
\hline
$x_3$ &$T_3$ & $y_I,y_4,y_5$ &$x_3+\underbrace{x_4+x_5}_{\text{side-information}}$  \\
\hline
$x_4$ & $\in V_{NI}$ & $y_2$ & $x_4+\underbrace{x_2}_{\text{side-information}}$ \\
\hline
$x_5$ & $\in V_{NI}$& $y_5$& $x_5+\underbrace{x_1}_{\text{side-information}}$ \\
\hline
\end{tabular}
\vspace{5pt}
\caption{Decoding of ICP described by Figure \ref{fig2}}
\label{table111}
\end{table}

\begin{figure}[t]
\centering
\includegraphics[scale=0.4]{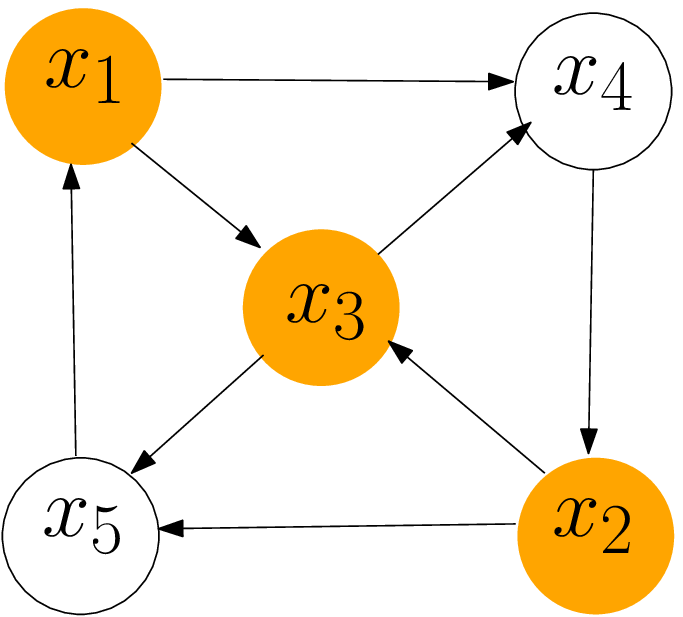}
\caption{Interlinked cycle structure with $V_I=\{x_1,x_2,x_3\}$.}
\label{ex1fig}
\end{figure}
\begin{figure}[t]
\centering
\includegraphics[scale=0.4]{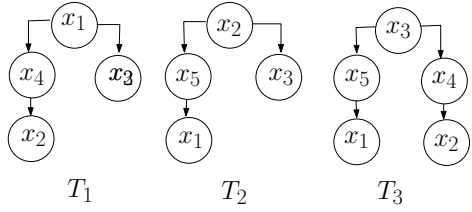}
\caption{Trees of inner vertices of IC structure given in Fig. \ref{ex1fig}.}
\label{ex1fig2}
\end{figure}
\end{example}

\begin{example}
\label{ex2}
Consider an SUICP with side-information graph $G$ given in Fig. \ref{ex2fig}. $G$ is an IC structure with $K=7,N=4$ and inner vertices $V_I=\{x_1,x_2,x_3,x_4\}$. Hence, for this side-information graph, we have $K-N+1=4$. An optimal length index code $\mathfrak{C}$ for this ICP obtained from \eqref{code1} and \eqref{code2} is 
$$\{\underbrace{\underbrace{x_1+x_2+x_3+x_4}_{\in V_I}}_{y_I},~\underbrace{\underbrace{x_5}_{\in V_{NI}}+\underbrace{x_4}_{N_{G}^+(x_5)}}_{y_5},~\underbrace{\underbrace{x_6}_{\in V_{NI}}+\underbrace{x_1+x_3+x_4}_{N_{G}^+(x_6)}\}}_{y_6}.$$

Trees $T_1,T_2,T_3$ and $T_4$ corresponding to the inner vertices $x_1,x_2,x_3$ and $x_3$ are given in Figure \ref{ex2fig2}. The decoding of each message symbol from $\mathfrak{C}$ is summarised in Table \ref{table112}. 

\begin{table}[ht]
\centering
\setlength\extrarowheight{3pt}
\begin{tabular}{|c|c|c|c|}
\hline
$x_k$ & Tree &$\gamma_k$ & $\tau_k$ \\
\hline
$x_1$ & $T_1$& $y_I,y_5$& $x_1+\underbrace{x_2+x_3+x_5}_{\text{side-information}}$ \\
\hline
$x_2$ & $T_2$ & $y_I,y_6$ & $x_2+\underbrace{x_1+x_3+x_6}_{\text{side-information}}$  \\
\hline
$x_3$ &$T_3$ & $y_I,y_6$ &$x_3+\underbrace{x_1+x_2+x_5}_{\text{side-information}}$  \\
\hline
$x_4$ & $T_4$ & $y_I$ & $x_4+\underbrace{x_1+x_2+x_3}_{\text{side-information}}$ \\
\hline
$x_5$ & $\in V_{NI}$& $y_5$& $x_5+\underbrace{x_4}_{\text{side-information}}$ \\
\hline
$x_6$ & $\in V_{NI}$& $y_6$& $x_6+\underbrace{x_1}_{\text{side-information}}$ \\
\hline
\end{tabular}
\vspace{5pt}
\caption{Decoding of ICP described by Figure \ref{fig2}}
\label{table112}
\end{table}

\begin{figure}[t]
\centering
\includegraphics[scale=0.4]{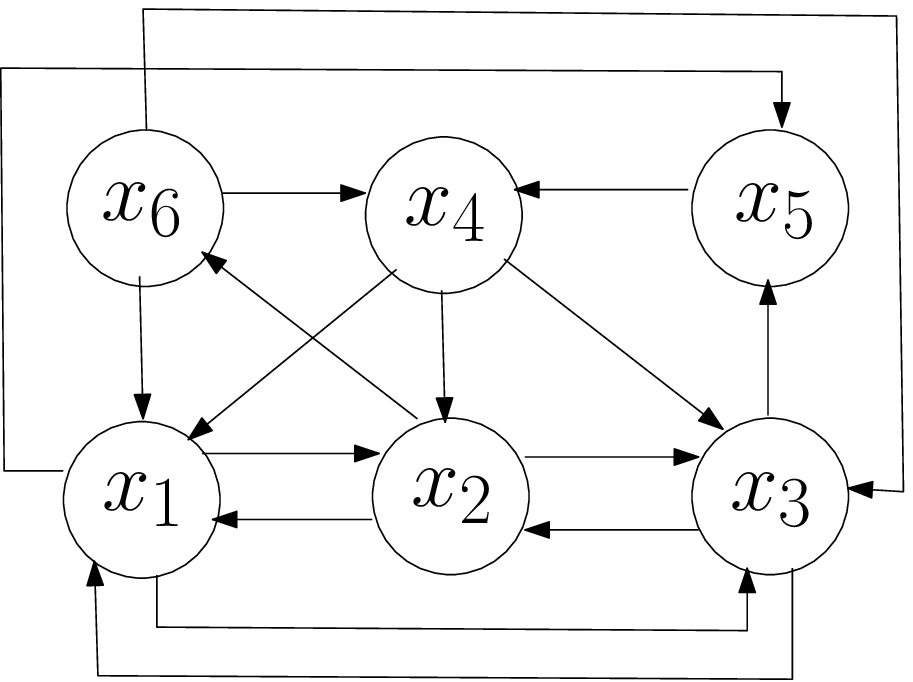}
\caption{Interlinked cycle structure with $V_I=\{x_1,x_2,x_3,x_4\}$.}
\label{ex2fig}
\end{figure}
\begin{figure}[t]
\centering
\includegraphics[scale=0.4]{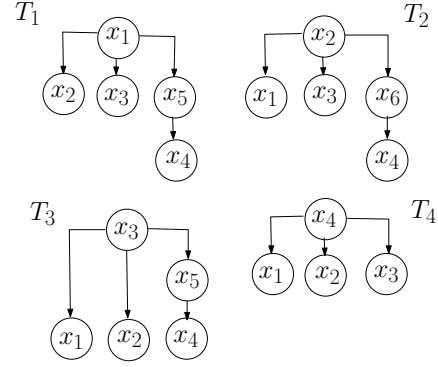}
\caption{Trees of inner vertices of IC structure given in Fig. \ref{ex2fig}.}
\label{ex2fig2}
\end{figure}
\end{example}

In \cite{VaR3}, we provided an addition to interlinked cycle structure class by providing optimal length
index codes for IC structures with one cycle among non-inner vertex set. We gave a modified code construction and modified decoding method for the IC structure with one cycle among non-inner vertex set. In \cite{VaR4}, we disproved the two conjectures given in \cite{ICC} regarding the optimality of IC structures.

\subsection{Motivation for overlapping interlinked cycle structure}
In \cite{VaR1}, we designed binary matrices of size $m \times n~(m>n)$ such that any $n$ adjacent rows of these matrices are linearly independent. We refer these matrices as Adjacent Independent Row (AIR) matrices. We used AIR matrices to give optimal length index codes for some symmetric index coding problems. In \cite{VaR2}, we give a low complexity decoding for the index coding problems which use AIR matrix as encoding matrix. The low complexity decoding method uses only a specific subset (less than the actual side-information available) of the side-information to decode a given message. If we consider side-information graphs, and obtain sub-graphs by only retaining the edges exploited for  low complexity decoding, we obtain large number of side-information graphs with known capacity and broadcast rate and which are not interlinked cycle structures. 

The below mentioned three examples are useful to understand the motivation behind studying overlapping interlinked cycle structures.
\subsection*{Motivating Example I}
Consider the side-information graph $G$ given in Figure \ref{fig2}. The broadcast rate of the index coding problem described by this side-information graph is three. To get an index code of length three, the side-information graph must have an inner vertex set with four inner vertices in it. But, no subset of size four of $\{x_1,x_2,\ldots,x_6\}$ satisfies the necessary conditions required for $V_I$. Hence, the side-information graph $G$ is not an interlinked cycle structure. 

In this paper we show that $G$ given in Figure \ref{fig2} is an overlapping interlinked cycle structure and give an index code with length MAIS(G).
\begin{figure}
\centering
\includegraphics[scale=0.4]{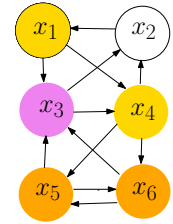}\\
\caption{Overlapping IC structure with capacity $\frac{1}{3}$.}
\label{fig2}
\end{figure}
\subsection*{Motivating Example II}
Consider the side-information graph $G$ given in Figure \ref{fig3}. The broadcast rate of the index coding problem described by this side-information graph is six. To get an index code of length six, the side-information graph must have an inner vertex set with five inner vertices in it. But, no subset of size five of $\{x_1,x_2,\ldots,x_{10}\}$ satisfies the necessary conditions required for $V_I$. Hence, the side-information graph $G$ is not an interlinked cycle structure. 

In this paper we show that $G$ given in Figure \ref{fig3} is an overlapping interlinked cycle structure give an index code with length MAIS(G).

\begin{figure}
\centering
\includegraphics[scale=0.4]{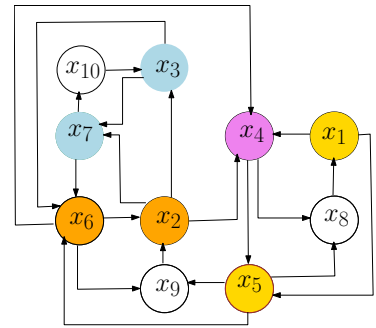}\\
\caption{Overlapping IC structure with capacity $\frac{1}{6}$.}
\label{fig3}
\end{figure}

\subsection*{Motivating Example III}
Consider the side-information graph $G$ given in Figure \ref{fig311}. The broadcast rate of the index coding problem described by this side-information graph is four. To get an index code of length four, the side-information graph must have an inner vertex set with seven inner vertices in it. But, no subset of size seven of $\{x_1,x_2,\ldots,x_{10}\}$ satisfies the necessary conditions required for $V_I$. Hence, the side-information graph $G$ is not an interlinked cycle structure. 

In this paper we show that $G$ given in Figure \ref{fig311} is an overlapping interlinked cycle structure give an index code with length MAIS(G).

\begin{figure}
\centering
\includegraphics[scale=0.5]{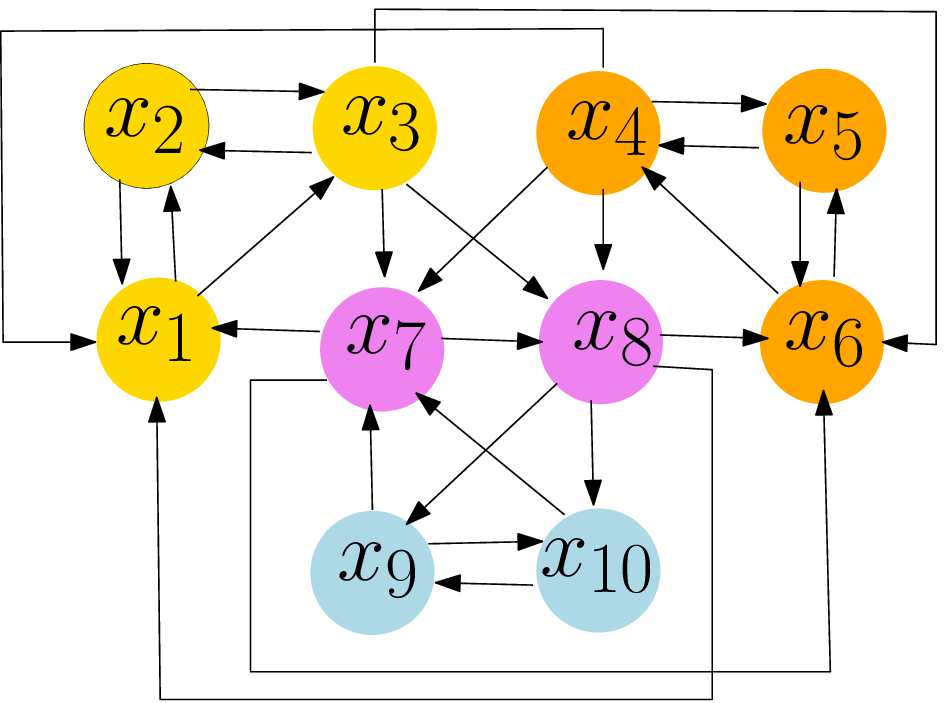}\\
\caption{Overlapping IC structure with capacity $\frac{1}{4}$.}
\label{fig311}
\end{figure}

\subsection{Contributions}
The contributions of this paper are summarized as below:
\begin{itemize}
\item  We extend the notion of interlinked cycle structure to define more generalised graph structures called overlapping interlinked cycle (OIC) structures. We give an index code for OIC structures whose length is equal to MAIS of OIC structure.
\end{itemize}

\section{Overlapping Cycle Structures}
In this section, we generalise the notion of interlinked cycle structure to define overlapping interlinked cycle structure. 

A tree is an undirected graph in which any two vertices are connected by exactly one path. That is, a tree is an acyclic connected graph. A polytree \cite{polytree} is a directed acyclic graph whose underlying undirected graph is a tree. In a polytree, there exits only one directed path from any vertex to any other vertex. 

A graph $G$ is called overlapping interlinked cycle structure if there exists a collection of subsets of vertices of $V(G)$ double indexed as $V_I^{(i,j)}$ for $i \in [1:d]$ and $j \in [1:w_i]$ such that the following conditions are satisfied.
\begin{condition}
\label{con1}
The vertex subsets $V_I^{(i,j)}$ for $i \in [1:d]$ and $j \in [1:w_i]$ form a polytree with edges existing only between a parent and child when there is a single common vertex between the parent and child. Also, the number of vertices in $V_I^{(i,j)}$ must be greater than the sum of  number of parents and children in the polytree. The polytree is shown in Figure \ref{fig43}.
\end{condition}

To define the second condition, we need to define some sets related to the vertex sets present in polytree. Let $V_I^{\text{Total}}=\cup_{i=0}^d \cup_{j=1}^{w_i} V_I^{(i,j)}$ and $V_{NI}=V(G)\setminus V_I^{\text{Total}}$. We refer vertices in $V_I^{\text{Total}}$ as inner vertices and the vertices in $V_{NI}$ as non-inner vertices. If the vertex set $V_I^{(i,j)}$ has $p$ number of parents, then, $V_I^{(i,j)}$ is having one common vertex with each of these $p$ parents. Let the set $\tilde{V}_I^{(i,j)}$ be the set after removing all the $p$ vertices from $V_I^{(i,j)}$ which are common with their $b$ parents. We have 
\begin{align*}
V_I^{\text{Total}}=\cup_{i=1}^d \cup_{j=1}^{w_i} V_I^{(i,j)}=\cup_{i=1}^d \cup_{j=1}^{w_i} \tilde{V}_I^{(i,j)}.
\end{align*}

If the vertex set $V_I^{(i,j)}$ has $c$ number of children, then, $V_I^{(i,j)}$ is having one common vertex with each of these $c$ children and $\tilde{V}_I^{(i,j)}$ consists of $c$ common vertices with its $c$ children and $|\tilde{V}_I^{(i,j)}|-c$ number of vertices which are not common to any other vertex set in the polytree. Let the $c$ children of $V_I^{(i,j)}$ in polytree be $V_I^{(i+1,j_1)},V_I^{(i+1,j_2)},\ldots,V_I^{(i+1,j_c)}$ and the corresponding common vertices be $x_{(i,j),j_1},x_{(i,j),j_2},\ldots,x_{(i,j),j_c}$ respectively.

Let $S^{(i,j),j_k}$ for $k \in [1:c]$ be the collection of all nodes $V_I^{(i^\prime,j^\prime)}$ in the polytree to which there exists a path from $V_I^{(i,j)}$ to $V_I^{(i^\prime,j^\prime)}$ through $V_I^{(i+1,j_k)}$. For every $V_I^{(i^\prime,j^\prime)} \in S^{(i,j),j_k}$, $i^\prime \in [i+1:d]$, there exists $i^\prime-i+1$ nodes of polytree present in this path including the first node $V_I^{(i,j)}$ and the last node ${V}_I^{(i^\prime,j^\prime)}$ and these $i^\prime-i+1$ nodes are connected by $i^\prime-i$ edges. Note that any two nodes in the polytree which are connected by an edge have a common vertex. Let this path be as given below.
\begin{align*}
&V_I^{(i,j)} \xrightarrow[x_{(i,j),j_k}]{}   V_I^{(i+1,j_k)} \xrightarrow[x_{(i+1,j_k),k_2}]{} V_I^{(i+2,k_2)} \xrightarrow[x_{(i+2,k_2),k_3}]{} \ldots \\& \xrightarrow[x_{(i^\prime-2,k_{i^\prime-i-2}),k_{i^\prime-i-1}}]{} V_I^{(i^\prime-1,k_{i^\prime-i-1})} \xrightarrow[x_{(i^\prime-1,k_{i^\prime-i-1}),k^\prime}]{k^\prime=j_{i^\prime-i}} V_I^{(i^\prime,j^\prime)}.
\end{align*}
Let 
\begin{align*}
&V_P^{(i,j),(i^\prime,j^\prime)}=V_I^{(i+1,j_k)} \bigcup \left(\bigcup\limits_{s=2}^{i^\prime-i} V_I^{(i+s,k_s)}\right)\setminus \\& \{x_{(i,j),j_k},x_{(i+1,j_k),k_2},x_{(i+2,k_2),k_3},\ldots,x_{(i^\prime-1,j_{i^\prime-i-1}),j^\prime}\}.
\end{align*}

That is, $V_P^{(i,j),(i^\prime,j^\prime)}$ is the union of the $i^\prime-i$ vertex sets present in the path from $V_I^{(i,j)}$ to $V_I^{(i^\prime,j^\prime)}$ excluding the vertex set $V_I^{(i,j)}$ and after removing $i^\prime-i$ common vertices.

\begin{condition}
\label{con2}
For every $\tilde{V}_I^{(i,j)}$, from every vertex in $\tilde{V}_I^{(i,j)}$, there should be only one path to every non-common vertex in that $\tilde{V}_I^{(i,j)}$ such that the path does not include any other inner vertices other than the first and last vertex. From every vertex in $\tilde{V}_I^{(i,j)}$, either there can be only one path to the common vertex $x_{(i,j),j_k}$ for $k \in [1:c]$ or there can be only one path to every vertex in $V_P^{(i,j),(i^\prime,j^\prime)}$ for any $V_I^{i^\prime,j^\prime}$ present in $S^{(i,j),j_k}$ such that the path does not include any other inner vertices other than the first and last vertex. All the paths mentioned in this condition are referred as $I$-paths in this paper.
\end{condition}

\begin{condition}
\label{con3}
For every inner vertex $x_{(i,j),k_1} \in V_I^{(i,j)}$, there exists no cycle in $G$ that includes $x_{(i,j),k_1}$ and vertices only from the set $V(G) \setminus V_I^{(i,j)}$. The graph $G$ should not have any cycle with only non-inner vertices in it.
\end{condition}
\begin{condition}
\label{con4}
Every non-inner vertex must be present in atleast one $I$-path. All the outgoing paths from a non-inner vertex terminate at the vertices of only one vertex set $V_I^{(i,j)}$ for $i \in [0:d]$ and $j \in [1:w_i]$.
\end{condition}
\begin{figure*}
\centering
\includegraphics[scale=0.55]{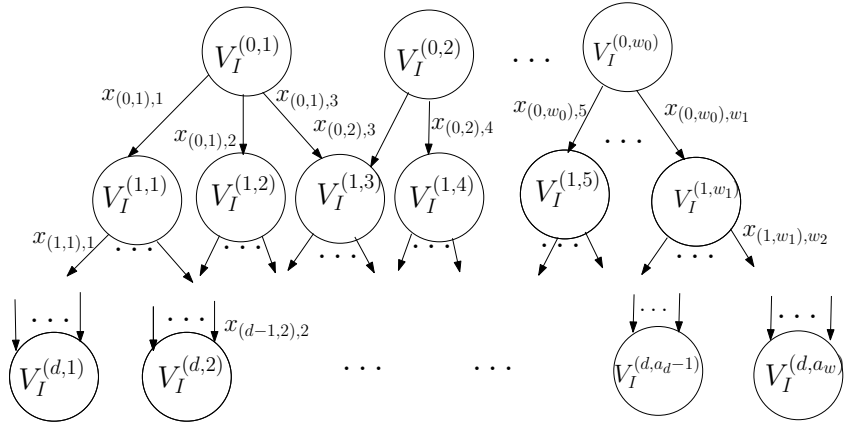}\\
\caption{Polytree structure of vertex subsets.}
\label{fig43}
\end{figure*}

In this paper, we refer the vertex subsets $V_I^{(i,j)}$ for $i \in [0:d]$ and $j \in [1:w_i]$ as semi-inner vertex sets. The following three examples illustrate Condition \ref{con2}.
\begin{example}
\label{ex5}
Consider the side-information $G$ given in Figure \ref{fig2}. In $G$, we have $V_I^{(0,1)}=\{x_1,x_3,x_4\}$ and $V_I^{(1,1)}=\{x_3,x_5,x_6\}$. The polytree structure of two semi-inner vertex sets are shown on Figure \ref{fig34}. In this graph $G$, the details of $I$-paths are mentioned in Table \ref{table2}. In $G$, the vertex $x_4$, instead of having an $I$-path to $x_3$, it has $I$-paths to $\{x_3,x_5,x_6\}\setminus \{x_3\}=\{x_5,x_6\}$.

\begin{table}[ht]
\centering
\setlength\extrarowheight{3pt}
\begin{tabular}{|c|c|c|c|}
\hline
$x_k$ &$I$-path & Depth at which & Depth at which\\
~~&~~& $I$-path originates & $I$-path terminates \\
\hline
$x_1$ & $x_3,x_4$ & 0 & 0 \\
\hline
$x_3$ & $x_1,x_4$ & 0 & 0  \\
\hline
$x_4$ & $x_1,\{x_5,x_6\}$ & 0 & 1 \\
\hline
$x_5$ & $x_3,x_6$ & 1 & 1 \\
\hline
$x_6$ & $x_3,x_5$ & 1 & 1  \\
\hline
\end{tabular}
\vspace{5pt}
\caption{$I$-paths present in Figure \ref{fig2}}
\label{table2}
\end{table}

\begin{figure}
\centering
\includegraphics[scale=0.4]{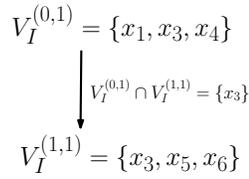}\\
\caption{Polytree of semi-inner vertex sets of OIC structure given in Figure \ref{fig2}.}
\label{fig34}
\end{figure}
\end{example}

\begin{example}
\label{ex6}
Consider the side-information graph given in Figure \ref{fig3}. In the graph $G$, we have $V_I^{(0,1)}=\{x_1,x_4,x_5\},V_I^{(1,1)}=\{x_2,x_4,x_6\}$ and $V_I^{(2,1)}=\{x_3,x_6,x_7\}$. The polytree structure of the three semi-inner vertex sets are shown on Figure \ref{fig35}. In $G$, the vertex $x_5$, instead of having an $I$-path to $x_4$, it has $I$-paths to $\{x_2,x_4,x_6\}\setminus \{x_4\}=\{x_2,x_6\}$. Similarly, the vertex $x_2$, instead of having an $I$-path to $x_6$, it has $I$-paths to $\{x_3,x_6,x_7\}\setminus \{x_6\}=\{x_3,x_7\}$.  The details of $I$-paths in $G$ are summarised in Table \ref{table3}.

\begin{table}[ht]
\centering
\setlength\extrarowheight{3pt}
\begin{tabular}{|c|c|c|c|}
\hline
$x_k$ &$I$-path & Depth at which & Depth at which\\
~~&~~& $I$-path originates & $I$-path terminates \\
\hline
$x_1$ & $x_4,x_5$& 0 & 0   \\
\hline
$x_4$ & $x_1,x_5$& 0 & 0  \\
\hline
$x_5$ & $x_1,\{x_2,x_6\}$& 0 & 1  \\
\hline
$x_2$ & $x_4,x_6$& 1& 1 \\
\hline
$x_6$ & $x_2,\{x_3,x_7\}$& 1 & 2  \\
\hline
$x_7$ & $x_4,x_3$& 2& 2 \\
\hline
$x_3$ & $x_4,x_7$& 2& 2 \\
\hline
\end{tabular}
\vspace{5pt}
\caption{$I$-paths present in Figure \ref{fig3}}
\label{table3}
\end{table}

\begin{figure}
\centering
\includegraphics[scale=0.4]{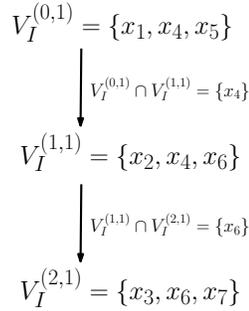}\\
\caption{Polytree of semi-inner vertex sets of OIC structure given in Figure \ref{fig3} and Figure \ref{fig351}.}
\label{fig35}
\end{figure}
\end{example}

\begin{example}
\label{ex61}
Consider the side-information given in Figure \ref{fig351}. In the graph $G$, we have $V_I^{(0,1)}=\{x_1,x_4,x_5\},V_I^{(1,1)}=\{x_2,x_4,x_6\}$ and $V_I^{(2,1)}=\{x_3,x_6,x_7\}$. The polytree structure of these side-information graph is shown in Figure \ref{fig35}. Note that the only difference between Figure \ref{fig3} and Figure \ref{fig351} is that the vertex $x_5$ have $I$-paths to $\{x_3,x_7\}$ instead of having an $I$-path to $x_6$. Hence, the $I$-path of $x_5$ is terminated at depth two instead of depth one. The details of $I$-paths are mentioned in Table \ref{table3}. 

\begin{figure}
\centering
\includegraphics[scale=0.4]{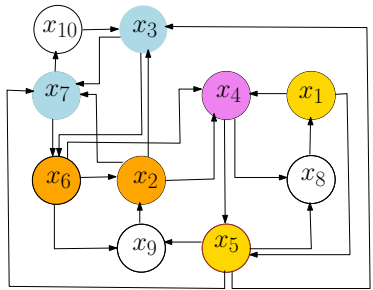}\\
\caption{Overlapping IC structure with capacity $\frac{1}{5}$.}
\label{fig351}
\end{figure}

\begin{table}[ht]
\centering
\setlength\extrarowheight{3pt}
\begin{tabular}{|c|c|c|c|}
\hline
$x_k$ &$I$-path & Depth at which & Depth at which\\
~~&~~& $I$-path originates & $I$-path terminates \\
\hline
$x_1$ & $x_4,x_5$ & 0 & 0 \\
\hline
$x_4$ & $x_1,x_5$& 0 & 0  \\
\hline
$x_5$ & $x_1,x_2,\{x_3,x_7\}$& 0 & 2  \\
\hline
$x_6$ & $x_2,x_4$& 1 & 1 \\
\hline
$x_2$ & $x_6,\{x_3,x_7\}$& 1& 2 \\
\hline
$x_7$ & $x_6,x_3$&2 & 2 \\
\hline
$x_3$ & $x_6,x_7$& 2&  2\\
\hline
\end{tabular}
\vspace{5pt}
\caption{$I$-paths present in Figure \ref{fig351}}
\label{table4}
\end{table}
\end{example}

The graphs given in Figure \ref{fig2}, Figure \ref{fig3} and Figure \ref{fig351} satisfy all the four conditions given above. Hence, they are overlapping interlinked cycle structures.

\subsection{Index code construction for Overlapping Interlinked Cycle Structure}
Consider an  index coding problem whose side-information graph is a Overlapping Interlinked Cycle (OIC) structure. in the following two steps, we give an index code of length $|V_{NI}|+s$. The index code comprises one code symbol for every semi-inner vertex set (total $s$ index code symbols for $s$ semi-inner vertex sets) and one index code symbol for every non-inner vertex ($|V_{NI}|$ index code symbols for $|V_{NI}|$ non-inner vertices).

\begin{itemize}
\item A code symbol is obtained by the bitwise XOR of messages present in the semi-inner vertex set $V_I^{(i,j)}$ for every $i \in [0:d]$ and every $j \in [1:w_i]$, i.e.,
\begin{align}
\label{code11}
y_I^{(i,j)}=\bigoplus_{k=1}^{|V_I^{(i,j)}|} x_{(i,j),k}.
\end{align}
\item For each $x_k \in V_{NI}$, a code symbol is obtained as given below.
\begin{align}
\label{code21}
y_{k}=x_k \bigoplus_{x_q \in N_{G}^+(x_k)} x_q.
\end{align}
where $N_{G}^+(x_k)$ is the out-neighborhood of $x_k$ in the IC structure $G$.
\end{itemize}
\subsection{Decoding procedure}
In this subsection, we give the decoding procedure for the index code constructed from \eqref{code11} and \eqref{code21} for OIC structures. 


To establish decoding procedure for OIC structures, we define the tree $T^{(i,j)}_k$ for every vertex $x_{(i,j),k} \in \tilde{V}^{(i,j)}_I$ for every $i \in [0:d]$ and $j \in [1:w_i]$. Let the $I$-paths originating from $x_{(i,j),k}$ pass through $b_k$ semi-inner vertex sets at depth $i+k$ for $k \in [0:t-i]$. Note that $b_0=1$ follows from the fact that $I$-paths originating from $x_{(i,j),k}$ pass through only one semi-inner vertex set $V^{(i,j)}_I$ at depth $i$. Let $t=\sum_{k=0}^{t-i} b_k$. From the definition of polytree, these $t$ semi-inner vertex sets are represented by $t$ nodes and these $t$ nodes are connected by $t-1$ directed edges in the polytree. Note that every edge in the semi-inner vertex set polytree represents one vertex in $G$ which is common to both the parent and child connected by this edge. Let $V_k^{(i,j)}$ be the union of $t$ semi-inner vertex sets after deleting the $t-1$ common vertices belonging to $t-1$ edges connecting the $t$ semi-inner vertex sets. That is, the cardinality of $V_k^{(i,j)}$ is $t$ less than that of the cardinality of union of $t$ semi-inner vertex sets. For $x_{(i,j),k} \in V_I^{(i,j)}$, because of the presence of $I$-paths from $x_{(i,j),k}$ to all other vertices in $V_k^{(i,j)}$, a directed rooted tree (denoted by $T_k^{(i,j)}$) in $G$ can be found with $x_{(i,j),k)}$ as the root vertex and all other inner vertices $V_k^{(i,j)} \setminus \{x_{(i,j),k}\}$ as the leaves.

From \eqref{code11}, all the message symbols in a semi-inner vertex set $V^{(i,j)}_I$ are encoded into one index code symbol $y_I^{(i,j)}$. Let $w^{(i,j)}_k$ be the XOR of $t$ index code symbols corresponding to the $t$ semi-inner vertex sets through which the $I$-paths originating from  $x_{(i,j),k}$ are pass through. That is, $w^{(i,j)}_k$ is the XOR of the message symbols present in $V_k^{(i,j)}$. 

Theorem \ref{dp} given below gives the decoding procedure for the index code constructed for OIC structures. The decoding procedure is same as that of the decoding procedure given by Thapa, Ong and Johnson in \cite{ICC} for IC structures except that the tree $T_k$ needs to be replaced by tree $T^{(i,j)}_k$ and $y_k$ needs to be replaced with $w^{(i,j)}_k$.
\begin{theorem}
\label{dp}
For any OIC structure, the index code constructed from \eqref{code11} and \eqref{code21} can be decoded by using the given below method.
\begin{itemize}
\item For $x_k \in V_{NI}$, the message $x_k$ can be decoded from $y_k$ given in \eqref{code21}.
\item The inner vertex $x_{(i,j),k} \in \tilde{V}_I^{(i,j)}$ is decoded by computing the XOR of all index code symbols corresponding to the non-leaf vertices at depth greater than zero in $T_k^{(i,j)}$ and $w_k^{(i,j)}$, where $T_k^{(i,j)}$ is the rooted tree with $x_{(i,j),k}$ as the root node and all other inner vertices in $V_k^{(i,j)}$ as the leaves.
\end{itemize}
\end{theorem}
\begin{proof}
Proof is given in appendix. 
\end{proof}
The following theorem establish the index coding capacity and broadcast rate of OIC structures. 

\begin{theorem}
The index coding capacity $C(G)$ of an OIC structure $G$ with $s$ semi-inner vertex set is given by 
\begin{align*}
C(G)=\frac{1}{|V_{NI}|+s}. 
\end{align*}
\end{theorem}
\begin{proof}
In \eqref{code11} and \eqref{code21}, we constructed an index code for OIC structure with length $|V_{NI}|+s$. Hence, we have
\begin{align}
\label{beta1}
\beta(G) \leq |V_{NI}|+s.
\end{align}

According to the definition of OIC structure, there exits atleast one vertex in each semi-inner vertex set that is not common to any other semi-inner vertex set. According to the definition of OIC structure, there exists no $I$-cycle with any inner vertex. Let this vertex be $x_{(i,j),k}$ in $V^{(i,j)}_I$ for $i \in [0:d]$ and $j \in [1:w_i]$. There exist $s$ vertices like this in $s$ semi-inner vertex sets one in each semi-inner vertex set. Consider the induced subgraph of these $s$ inner vertices and all non-inner vertices. Let this induced subgraph be $G_I$. According to the definition of OIC structure, there are no cycles among non-inner vertex set. Hence, $G_I$ is acyclic. We have
\begin{align}
\label{beta2}
\text{MAIS}(G)\geq |V_{NI}|+s.
\end{align}
But MAIS(G) gives the lowerbound on $\beta(G)$. Hence, from \eqref{beta1} and \eqref{beta2}, we have 
\begin{align*}
\beta(G)=|V_{NI}|+s.
\end{align*}
The index coding capacity is the reciprocal of the broadcast rate. This completes the proof.
\end{proof}

In the OIC structures, if the number of semi-inner vertex sets is one ($s=1$), then the OIC structure becomes an IC structure. Hence, when $s=1$, the encoding, decoding and optimality results in this paper exactly match the results given by Thapa, Ong and Johnson in \cite{ICC}.

The following six examples illustrates the encoding and decoding of OIC structures. 

\begin{example}
\label{ex99}
Consider the index coding problem with side-information graph given in Figure \ref{fig2}. In $G$, we have $V_I^{(0,1)}=\{x_1,x_3,x_4\}$ and $V_I^{(1,1)}=\{x_3,x_5,x_6\}$. The polytree structure of two inner vertex sets are shown on Figure \ref{fig34}. For $G$, the index code obtained from \eqref{code11} and \eqref{code21} is given below.
\begin{align*}
\mathfrak{C}=\{\underbrace{\underbrace{x_1+x_3+x_4}_{\in V_I^{(0,1)}}}_{y_I^{(0,1)}},~\underbrace{\underbrace{x_3+x_5+x_6}_{\in V_I^{(1,1)}}}_{y_I^{(1,1)}},~\underbrace{\underbrace{x_2}_{\in V_{NI}}+\underbrace{x_1}_{N_{G}^+(x_2)}\}}_{y_2}.
\end{align*}

Trees $T^{(0,1)}_1,T^{(0,1)}_3,T^{(0,1)}_4,T^{(1,1)}_5$ and $T^{(1,1)}_6$ corresponding to the inner vertices $x_1,x_3,x_4,x_5$ and $x_6$ are given in Figure \ref{fig63}. The decoding of each message symbol from $\mathfrak{C}$ is summarised in Table \ref{table11}. 

\begin{table}[ht]
\centering
\setlength\extrarowheight{3pt}
\begin{tabular}{|c|c|c|c|}
\hline
$x_k$ & Tree &$\gamma_k$ & $\tau_k$ \\
\hline
$x_1$ & $T_1$& $y_I^{(0,1)}$& $x_1+\underbrace{x_3+x_4}_{\text{side-information}}$ \\
\hline
$x_2$ & $\in V_{NI}$ & $y_2$ & $x_2+\underbrace{x_1}_{\text{side-information}}$  \\
\hline
$x_3$ &$T_3$ & $y_I^{(0,1)},y_2$ &$x_3+\underbrace{x_2+x_4}_{\text{side-information}}$  \\
\hline
$x_4$ & $T_4$ & $y_I^{(0,1)},y_I^{(1,1)},y_2$ & $x_4+\underbrace{x_2+x_5+x_6}_{\text{side-information}}$ \\
\hline
$x_5$ & $T_5$& $y_I^{(1,1)}$& $x_5+\underbrace{x_3+x_6}_{\text{side-information}}$ \\
\hline
$x_6$ & $T_6$ & $y_I^{(1,1)}$& $x_6+\underbrace{x_3+x_5}_{\text{side-information}}$ \\
\hline
\end{tabular}
\vspace{5pt}
\caption{Decoding of ICP described by Figure \ref{fig2}}
\label{table11}
\end{table}
\begin{figure}
\centering
\includegraphics[scale=0.4]{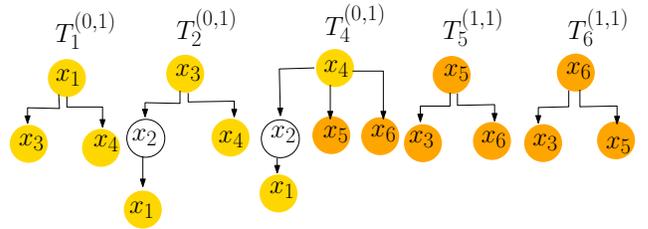}\\
\caption{Trees of inner vertices in Figure \ref{fig2}.}
\label{fig63}
\end{figure}
\end{example}

\begin{example}
\label{ex100}
Consider the index coding problem with side-information graph given in Figure \ref{fig3}. In the graph $G$, we have $V_I^{(0,1)}=\{x_1,x_4,x_5\},V_I^{(1,1)}=\{x_2,x_4,x_6\}$ and $V_I^{(2,1)}=\{x_3,x_6,x_7\}$. The polytree structure of the three semi-inner vertex sets are shown on Figure \ref{fig35}. The index code obtained from \eqref{code11} and \eqref{code21} is given below.
\begin{align*}
\mathfrak{C}=\{&\underbrace{\underbrace{x_1+x_4+x_5}_{\in V_I^{(0,1)}}}_{y_I^{(0,1)}},~\underbrace{\underbrace{x_2+x_4+x_6}_{\in V_I^{(1,1)}}}_{y_I^{(1,1)}},~\underbrace{\underbrace{x_3+x_6+x_{7}}_{\in V_I^{(2,1)}}}_{y_I^{(2,1)}},\\&\underbrace{\underbrace{x_8}_{\in V_{NI}}+\underbrace{x_1}_{N_{G}^+(x_8)}}_{y_8},~\underbrace{\underbrace{x_9}_{\in V_{NI}}+\underbrace{x_2}_{N_{G}^+(x_9)}}_{y_9},~\underbrace{\underbrace{x_{10}}_{\in V_{NI}}+\underbrace{x_3}_{N_{G}^+(x_{10})}}_{y_{10}}\}.
\end{align*}

The seven trees corresponding to the seven inner vertices are given in Figure \ref{fig64}. The decoding of each message symbol from $\mathfrak{C}$ is summarised in Table \ref{table12}. 

\begin{table}[ht]
\centering
\setlength\extrarowheight{3pt}
\begin{tabular}{|c|c|c|c|}
\hline
$x_k$ & Tree &$\gamma_k$ & $\tau_k$ \\
\hline
$x_1$ & $T_1$& $y_I^{(0,1)}$& $x_1+\underbrace{x_4+x_5}_{\text{side-information}}$ \\
\hline
$x_2$ & $T_2$ & $y_I^{(1,1)},y_I^{(2,1)}$ & $x_2+\underbrace{x_3+x_4+x_7}_{\text{side-information}}$  \\
\hline
$x_3$ &$T_3$ & $y_I^{(2,1)}$ &$x_3+\underbrace{x_6+x_7}_{\text{side-information}}$  \\
\hline
$x_4$ & $T_4$ & $y_I^{(0,1)},y_8$ & $x_4+\underbrace{x_5+x_8}_{\text{side-information}}$ \\
\hline
$x_5$ & $T_5$& $y_I^{(0,1)},y_I^{(1,1)},$& $x_5+\underbrace{x_6+x_8+x_9}_{\text{side-information}}$ \\
~ & ~& $y_8,~~y_9$& ~ \\
\hline
$x_6$ & $T_6$ & $y_I^{(1,1)}$& $x_6+\underbrace{x_2+x_4}_{\text{side-information}}$ \\
\hline
$x_7$ & $T_7$ & $y_I^{(2,1)}$& $x_7+\underbrace{x_4+x_{11}}_{\text{side-information}}$ \\
\hline
$x_8$ & $\in V_{NI}$ & $y_8$& $x_8+\underbrace{x_1}_{\text{side-information}}$ \\
\hline
$x_9$ & $\in V_{NI}$ & $y_9$& $x_9+\underbrace{x_2}_{\text{side-information}}$ \\
\hline
$x_{10}$ & $\in V_{NI}$ & $y_{10}$& $x_{10}+\underbrace{x_3}_{\text{side-information}}$ \\
\hline
\end{tabular}
\vspace{5pt}
\caption{Decoding of ICP described by Figure \ref{fig3}}
\label{table12}
\end{table}
\begin{figure}
\centering
\includegraphics[scale=0.4]{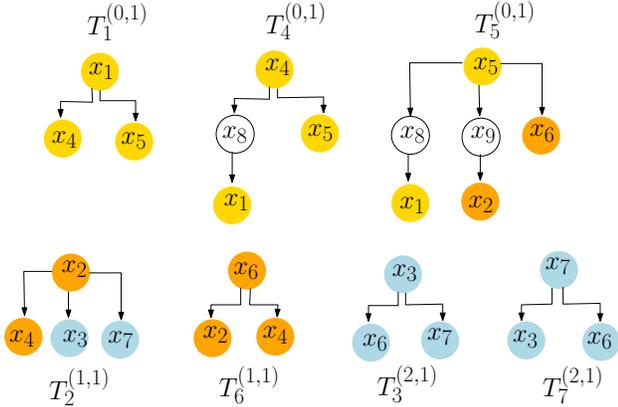}\\
\caption{Trees of inner vertices of Figure \ref{fig3}.}
\label{fig64}
\end{figure}
\end{example}

\begin{example}
\label{ex101}
Consider the index coding problem with side-information graph given in Figure \ref{fig351}. The index code for this index coding problem is same as that of the index code for the index coding problem described by the side-information graph Figure \ref{fig3}. 

The seven trees corresponding to the seven inner vertices are given in Figure \ref{fig66}. The decoding of each message symbol is exactly same as that of the decoding of each message symbol in Example \ref{ex100}. The decoding at each receiver is summarised in Table \ref{table12}. 

\begin{figure}
\centering
\includegraphics[scale=0.4]{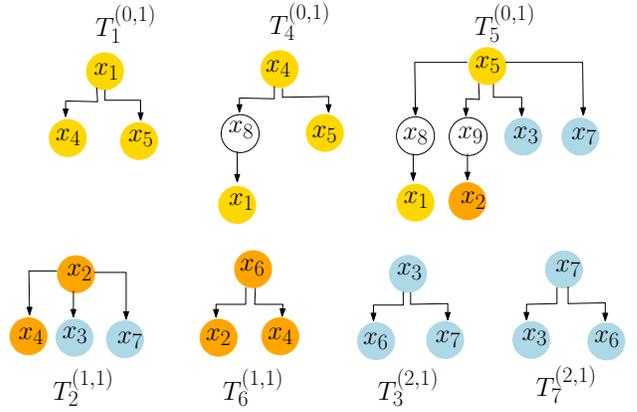}\\
\caption{Trees of inner vertices in Figure \ref{fig351}.}
\label{fig66}
\end{figure}
\end{example}

\begin{example}
\label{ex102}
Consider the side-information given in Figure \ref{fig31}. In this graph, we have $V_I^{(0,1)}=\{x_1,x_3,x_5,x_6\}$ and $V_I^{(1,1)}=\{x_2,x_4,x_5,x_7\}$. The polytree structure of the semi-inner vertex sets are given in Figure \ref{fig36}. The index code obtained from \eqref{code11} and \eqref{code21} is given below.
\begin{align*}
\mathfrak{C}=\{&\underbrace{\underbrace{x_1+x_3+x_5+x_6}_{\in V_I^{(0,1)}}}_{y_I^{(0,1)}},~\underbrace{\underbrace{x_2+x_4+x_5+x_7}_{\in V_I^{(1,1)}}}_{y_I^{(1,1)}},\\&\underbrace{x_8+x_1}_{y_8},~\underbrace{x_9+x_2}_{y_9},~\underbrace{x_{10}+x_3}_{y_{10}},~\underbrace{x_{11}+x_4}_{y_{11}},~\underbrace{x_{12}+x_5}_{y_{12}}\}.
\end{align*}

\begin{figure}
\centering
\includegraphics[scale=0.4]{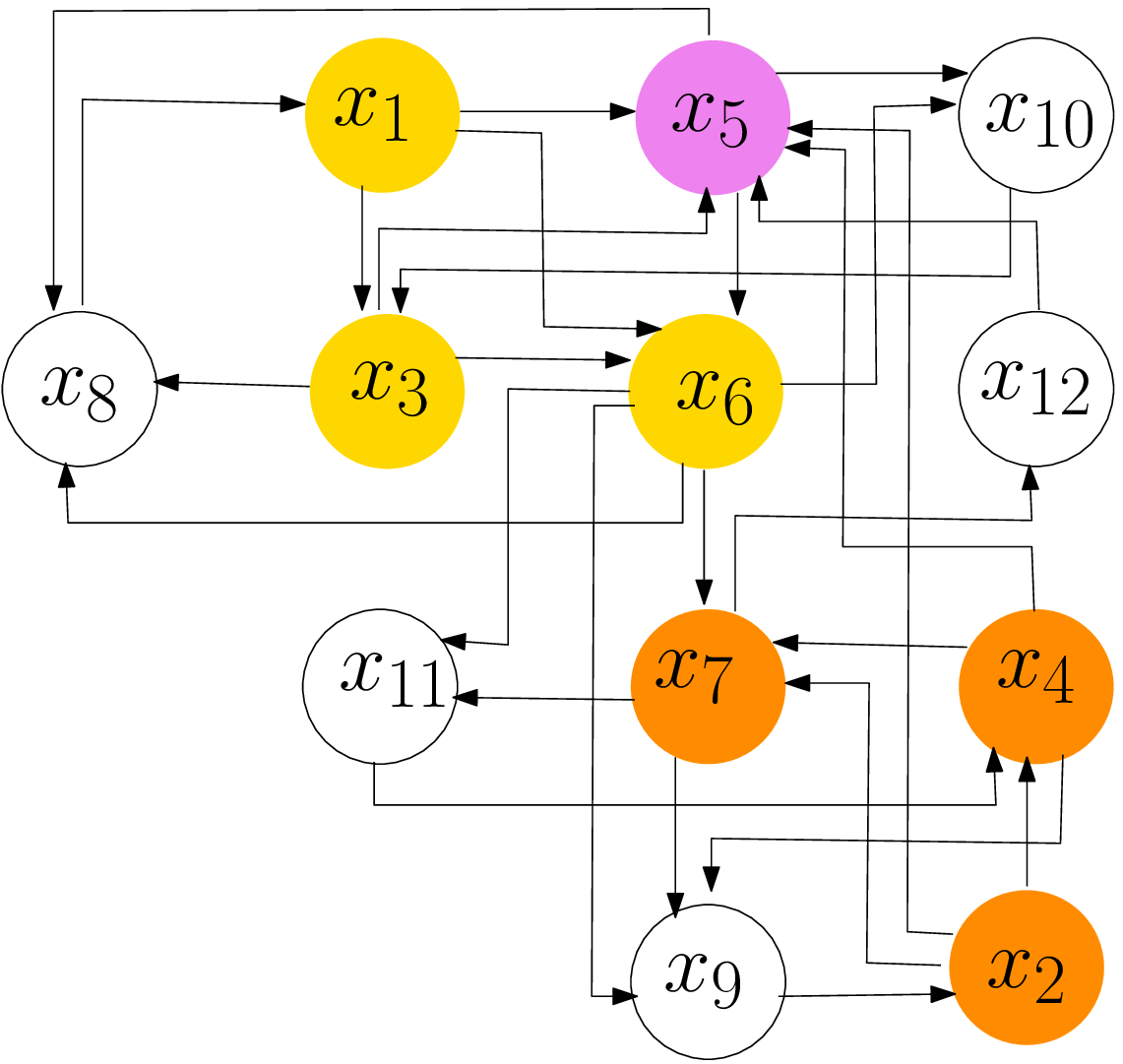}\\
\caption{Overlapping IC structure with capacity $\frac{1}{6}$.}
\label{fig31}
\end{figure}

\begin{figure}
\centering
\includegraphics[scale=0.4]{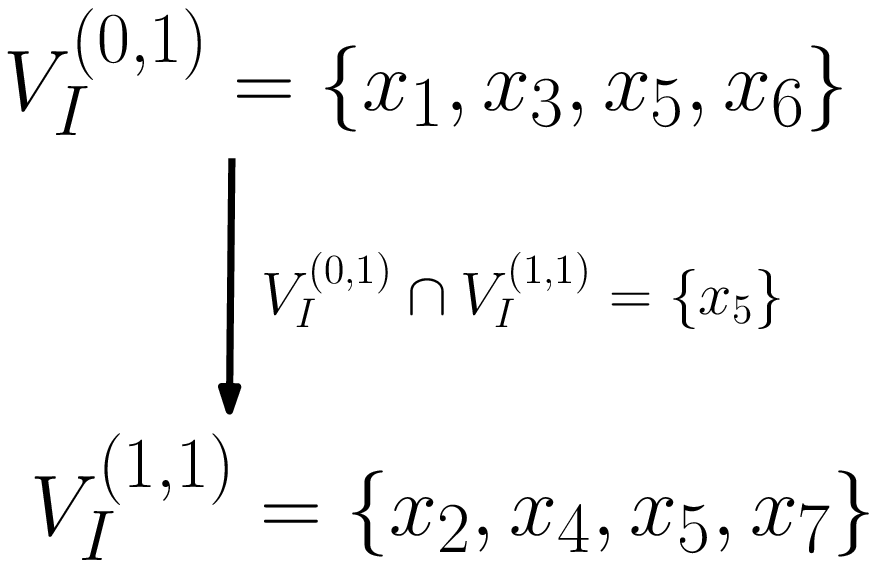}\\
\caption{Polytree of semi-inner vertex sets of OIC structure given in Figure \ref{fig31}.}
\label{fig36}
\end{figure}

\begin{figure}
\centering
\includegraphics[scale=0.4]{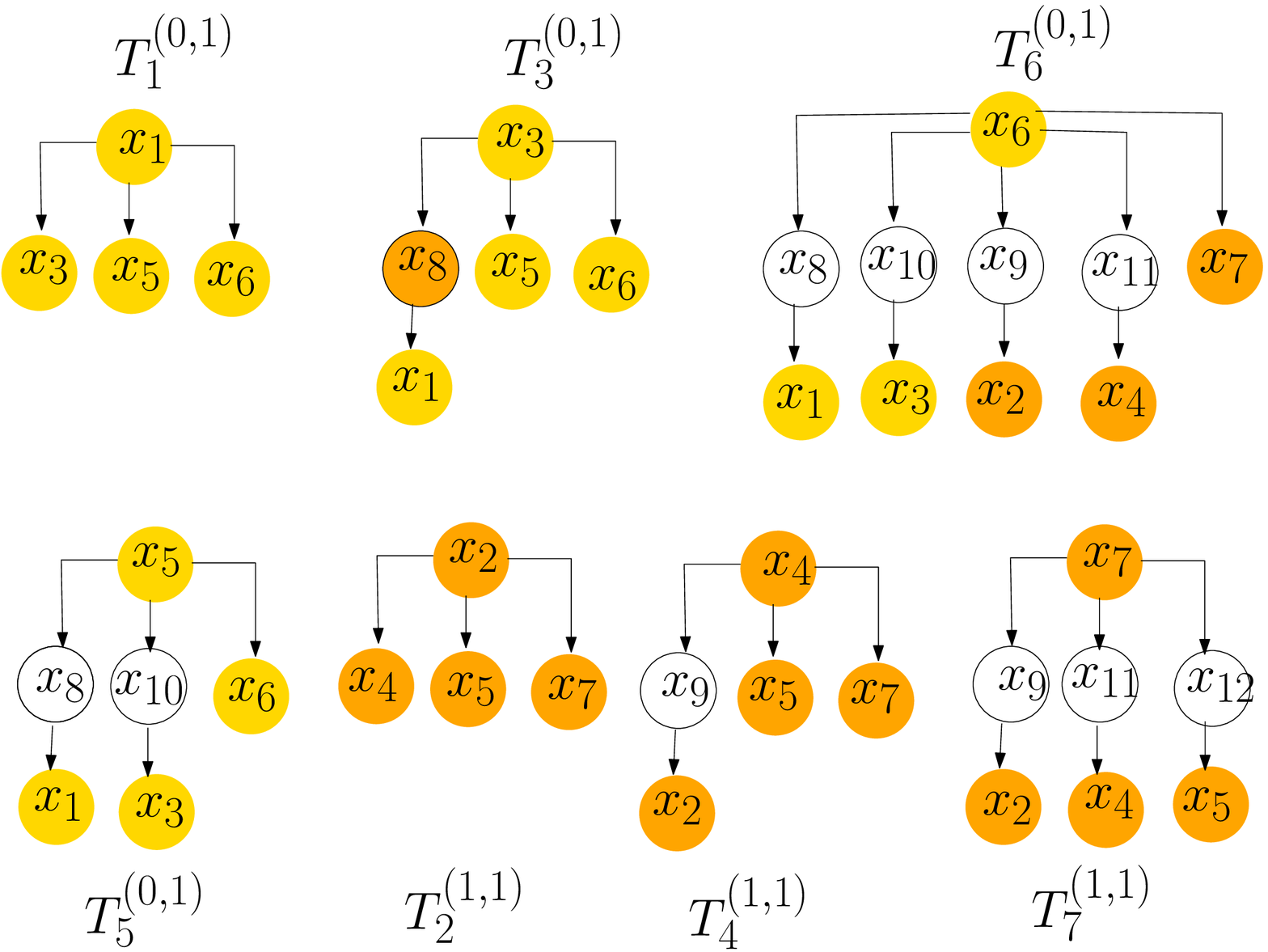}\\
\caption{Trees of inner vertices in Figure \ref{fig31}.}
\label{fig661}
\end{figure}

The seven trees corresponding to the seven inner vertices are given in Figure \ref{fig661}. The decoding of each message symbol from $\mathfrak{C}$ is summarised in Table \ref{table14}. 

\begin{table}[ht]
\centering
\setlength\extrarowheight{3pt}
\begin{tabular}{|c|c|c|c|}
\hline
$x_k$ & Tree &$\gamma_k$ & $\tau_k$ \\
\hline
$x_1$ & $T_1$& $y_I^{(0,1)}$& $x_1+\underbrace{x_3+x_5+x_6}_{\text{side-information}}$ \\
\hline
$x_2$ & $T_2$ & $y_I^{(1,1)}$ & $x_2+\underbrace{x_4+x_5+x_7}_{\text{side-information}}$  \\
\hline
$x_3$ &$T_3$ & $y_I^{(0,1)},y_8$ &$x_3+\underbrace{x_8+x_5+x_6}_{\text{side-information}}$  \\
\hline
$x_4$ & $T_4$ & $y_I^{(1,1)},y_9$ & $x_4+\underbrace{x_9+x_5+x_7}_{\text{side-information}}$ \\
\hline
$x_5$ & $T_5$& $y_I^{(0,1)},y_8,y_{10}$& $x_5+\underbrace{x_8+x_{10}+x_6}_{\text{side-information}}$ \\
\hline
$x_6$ & $T_6$ & $y_I^{(0,1)},y_I^{(1,1)},y_8,$& $x_6+$ \\
 & & $y_9,~y_{10},~y_{1,1}$ & $\underbrace{x_6+x_8+x_9+x_{10}+x_{11}}_{\text{side-information}}$\\
\hline
$x_7$ & $T_7$& $y_I^{(1,1)},y_9,y_{10},y_{11}$& $x_7+\underbrace{x_9+x_{11}+x_{12}}_{\text{side-information}}$ \\
\hline
$x_8$ & $\in V_{NI}$ & $y_8$ & $x_8+\underbrace{x_1}_{\text{side-information}}$  \\
\hline
$x_9$ & $\in V_{NI}$  & $y_9$ &$x_9+\underbrace{x_2}_{\text{side-information}}$  \\
\hline
$x_{10}$ &$\in V_{NI}$ & $y_{10}$ & $x_{10}+\underbrace{x_3}_{\text{side-information}}$ \\
\hline
$x_{11}$ & $\in V_{NI}$& $y_{11}$& $x_{11}+\underbrace{x_4}_{\text{side-information}}$ \\
\hline
$x_{12}$ & $\in V_{NI}$ & $y_{12}$& $x_{12}+\underbrace{x_5}_{\text{side-information}}$ \\
\hline
\end{tabular}
\vspace{5pt}
\caption{Decoding of ICP described by Figure \ref{fig31}}
\label{table14}
\end{table}
\end{example}

\begin{example}
\label{ex103}
Consider the index coding problem described by the side-information given in Figure \ref{fig39}. In this graph, we have $V_I^{(0,1)}=\{x_1,x_2,x_3,x_4\},V_I^{(1,1)}=\{x_3,x_5,x_6\}$ and $V_I^{(1,2)}=\{x_5,x_7,x_8\}$. The polytree structure of the semi-inner vertex sets are given in Figure \ref{fig392}. The index code obtained from \eqref{code11} and \eqref{code21} is given below.
\begin{align*}
\mathfrak{C}=\{&\underbrace{\underbrace{x_1+x_2+x_3+x_4}_{\in V_I^{(0,1)}}}_{y_I^{(0,1)}},~\underbrace{\underbrace{x_3+x_5+x_6}_{\in V_I^{(1,1)}}}_{y_I^{(1,1)}},,~\underbrace{\underbrace{x_5+x_7+x_8}_{\in V_I^{(1,1)}}}_{y_I^{(1,2)}}\}.
\end{align*} 
\begin{figure}
\centering
\includegraphics[scale=0.5]{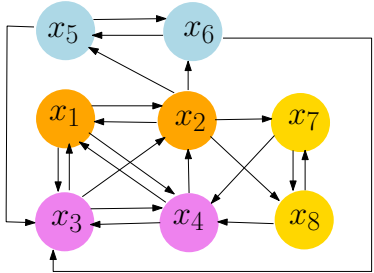}\\
\caption{Overlapping IC structure with capacity $\frac{1}{3}$.}
\label{fig39}
\end{figure}

\begin{figure}
\centering
\includegraphics[scale=0.5]{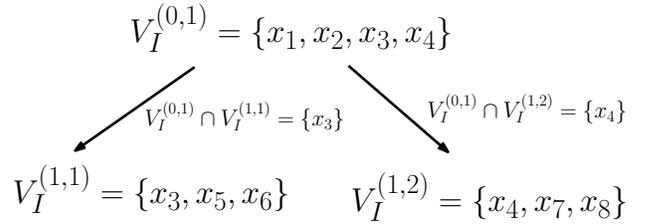}\\
\caption{Polytree of OIC given in Figure \ref{fig39}.}
\label{fig392}
\end{figure}

The eight trees corresponding to the eight inner vertices are given in Figure \ref{fig38}. The decoding of each message symbol from $\mathfrak{C}$ is summarised in Table \ref{table15}. 

\begin{figure}
\centering
\includegraphics[scale=0.4]{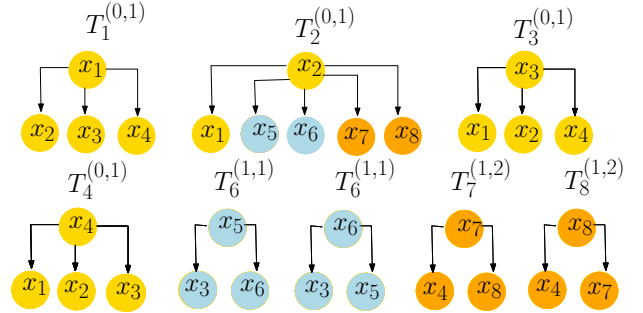}\\
\caption{Trees of inner vertices of Figure \ref{fig39}.}
\label{fig38}
\end{figure}

\begin{table}[ht]
\centering
\setlength\extrarowheight{3pt}
\begin{tabular}{|c|c|c|c|}
\hline
$x_k$ & Tree &$\gamma_k$ & $\tau_k$ \\
\hline
$x_1$ & $T_1$& $y_I^{(0,1)}$& $x_1+\underbrace{x_2+x_3+x_4}_{\text{side-information}}$ \\
\hline
$x_2$ & $T_2$ & $y_I^{(0,1)},y_I^{(1,1)},y_I^{(1,2)}$ & $x_2+\underbrace{x_1+x_5+x_6+x_7+x_8}_{\text{side-information}}$  \\
\hline
$x_3$ &$T_3$ & $y_I^{(0,1)}$ &$x_3+\underbrace{x_1+x_2+x_4}_{\text{side-information}}$  \\
\hline
$x_4$ & $T_4$ & $y_I^{(0,1)}$ & $x_4+\underbrace{x_1+x_2+x_3}_{\text{side-information}}$ \\
\hline
$x_5$ & $T_5$& $y_I^{(1,1)}$& $x_5+\underbrace{x_3+x_6}_{\text{side-information}}$ \\
\hline
$x_6$ & $T_6$ & $y_I^{(1,1)}$& $x_6+\underbrace{x_3+x_5}_{\text{side-information}}$ \\
\hline
$x_7$ & $T_7$& $y_I^{(1,2)}$& $x_7+\underbrace{x_4+x_{8}}_{\text{side-information}}$ \\
\hline
$x_8$ & $T_8$ & $y_I^{(1,2)}$ & $x_8+\underbrace{x_4+x_7}_{\text{side-information}}$  \\
\hline
\end{tabular}
\vspace{5pt}
\caption{Decoding of ICP described by Figure \ref{fig39}}
\label{table15}
\end{table}
\end{example}

\begin{example}
\label{ex104}
Consider the index coding problem described by the side-information given in Figure \ref{fig311}. The graph $G$ is an overlapping interlinked cycle structure with four inner vertex sets $\{x_1,x_2,x_3\},\{x_4,x_5,x_6\},\{x_1,x_6,x_7,x_8\}$ and $\{x_4,x_7,x_{8}\}$. The polytree structure of the semi-inner vertices is given in Figure \ref{fig352}. The index code obtained from \eqref{code11} and \eqref{code21} is given below.
\begin{align*}
\mathfrak{C}=\{&\underbrace{x_1+x_2+x_3}_{y_I^{(0,1)}},~\underbrace{x_4+x_5+x_6}_{y_I^{(1,1)}},,~\underbrace{x_3+x_4++x_7+x_8}_{y_I^{(1,2)}},\\&\underbrace{x_7+x_9+x_{10}}_{y_I^{(2,1)}}\}.
\end{align*} 

\begin{figure}
\centering
\includegraphics[scale=0.4]{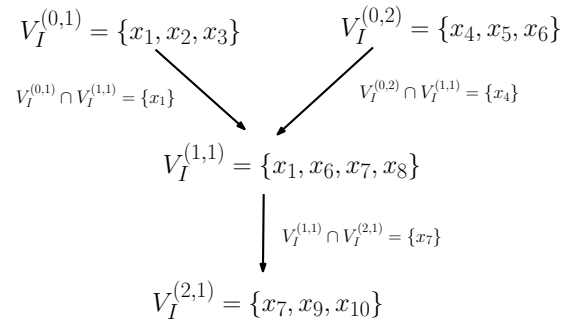}\\
\caption{Polytree of semi-inner vertex sets of OIC structure given in Figure \ref{fig311}.}
\label{fig352}
\end{figure}
\end{example}

\begin{table}[ht]
\centering
\setlength\extrarowheight{3pt}
\begin{tabular}{|c|c|c|c|}
\hline
$x_k$ & Tree &$\gamma_k$ & $\tau_k$ \\
\hline
$x_1$ & $T_1$& $y_I^{(0,1)}$& $x_1+\underbrace{x_2+x_3}_{\text{side-information}}$ \\
\hline
$x_2$ & $T_2$ & $y_I^{(0,1)}$ & $x_2+\underbrace{x_1+x_3}_{\text{side-information}}$  \\
\hline
$x_3$ &$T_3$ & $y_I^{(0,1)},y_I^{(1,1)}$ &$x_3+\underbrace{x_2+x_6+x_7+x_8}_{\text{side-information}}$  \\
\hline
$x_4$ & $T_4$ & $y_I^{(0,2)},y_I^{(1,1)}$ & $x_4+\underbrace{x_1+x_5+x_7+x_8}_{\text{side-information}}$ \\
\hline
$x_5$ & $T_5$& $y_I^{(0,2)}$& $x_5+\underbrace{x_4+x_6}_{\text{side-information}}$ \\
\hline
$x_6$ & $T_6$ & $y_I^{(0,2)}$& $x_6+\underbrace{x_4+x_5}_{\text{side-information}}$ \\
\hline
$x_7$ & $T_7$& $y_I^{(1,1)}$& $x_7+\underbrace{x_1+x_6+x_8}_{\text{side-information}}$ \\
\hline
$x_8$ & $T_8$ & $y_I^{(1,1)}$ & $x_8+\underbrace{x_1+x_6+x_9+x_{10}}_{\text{side-information}}$  \\
\hline
$x_9$ & $T_9$& $y_I^{(2,1)}$& $x_7+\underbrace{x_7+x_{10}}_{\text{side-information}}$ \\
\hline
$x_{10}$ & $T_{10}$ & $y_I^{(2,1)}$ & $x_{10}+\underbrace{x_7+x_9}_{\text{side-information}}$  \\
\hline
\end{tabular}
\vspace{5pt}
\caption{Decoding of ICP given in Example \ref{ex104}.}
\label{table16}
\end{table}
The ten trees corresponding to the all inner vertices are given in Figure \ref{fig381}. The decoding of each message symbol from $\mathfrak{C}$ is summarised in Table \ref{table16}. 

\begin{figure}
\centering
\includegraphics[scale=0.4]{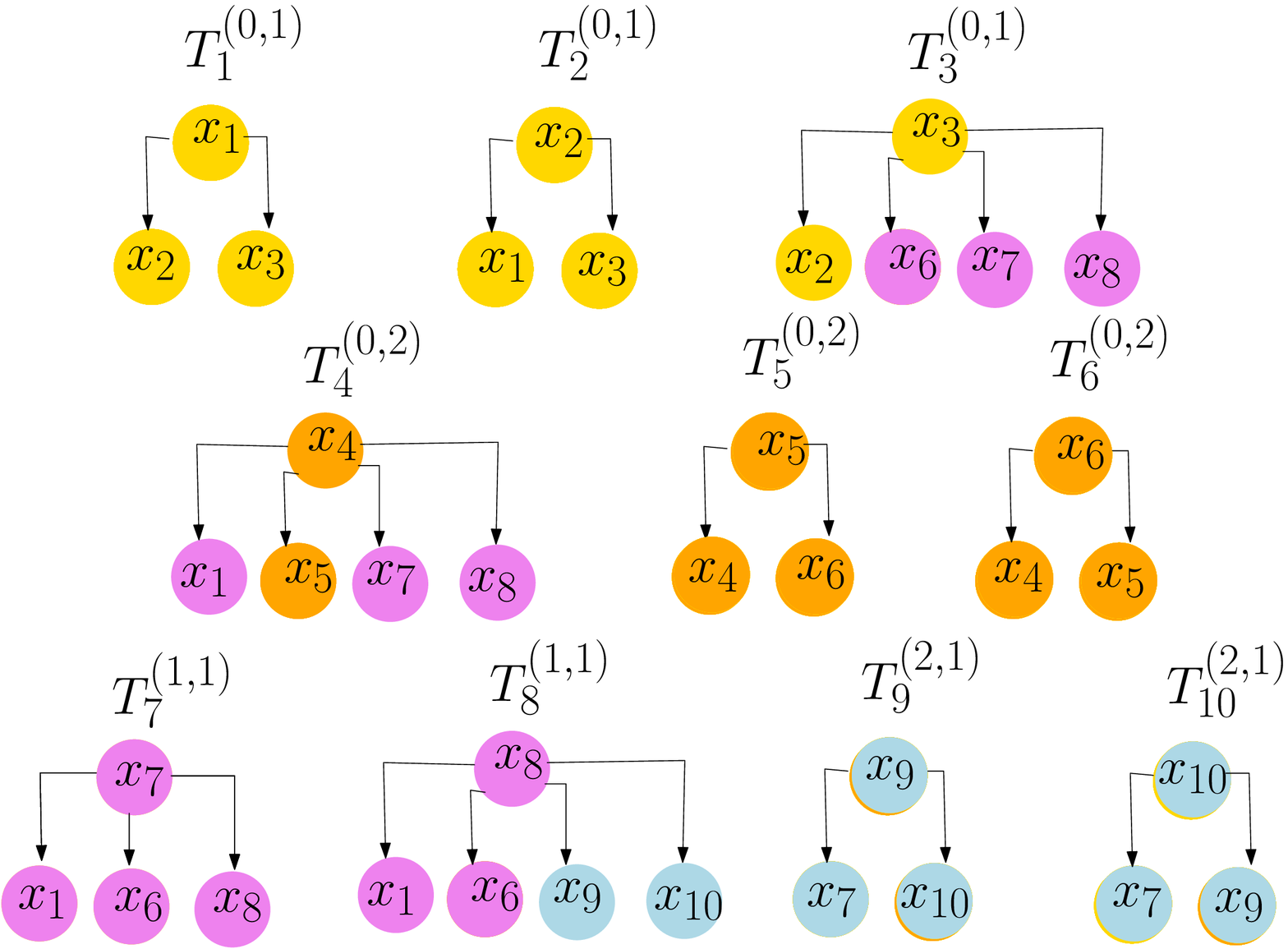}\\
\caption{Trees of inner vertices of Figure \ref{fig311}.}
\label{fig381}
\end{figure}

\section*{Acknowledgement}
This work was supported partly by the Science and Engineering Research Board (SERB) of Department of Science and Technology (DST), Government of India, through J.C. Bose National Fellowship to B. Sundar Rajan.

\begin{appendix}
\begin{center}
Proof of Theorem \ref{dp} 
\end{center}

To prove Theorem \ref{dp}, we need certain properties of the tree $T_k^{(i,j)}$ of inner vertex $x_{(i,j),k} \in V_I^{(i,j)}$ for every $i \in [0:d]$ and $j \in [1:w_i]$. In the following lemma, we prove important property of $T_k^{(i,j)}$.
\begin{lemma}
\label{lemma1}
For any non-inner vertex $x_{s} \in V(T_k^{(i,j)})$ for $i \in [0:d]$ and $j \in [1:w_i]$, the out-neighborhood of $x_s$ is same in $T_k^{(i,j)}$ and $G$.
\end{lemma}
\begin{proof}
 To prove this lemma, first we prove that for any non inner vertex $x_{s} \in V(T_k^{(i,j)})\cap V(T_{k^\prime}^{(i^\prime, j^\prime)})$, the out-neighborhood of vertex $x_{s}$ is same in both the trees $T_k^{(i,j)}$ and $T_{k^\prime}^{(i^\prime, j^\prime)}$. That is, we prove that 
\begin{align}
\label{on}
N^+_{T^{(i,j)}_k}(x_s)=N^+_{T^{(i^\prime, j^\prime)}_{k^\prime}}(x_s).
\end{align}
where $N_{T^{(i,j)}_k}(x_s)$ and $N_{T^{(i^\prime, j^\prime)}_{k^\prime}}(x_s)$ are the out-neighborhoods of vertex $x_{s}$ in $T_k^{(i,j)}$ and $T_{k^\prime}^{(i^\prime, j^\prime)}$ respectively.

Let $L_{T^{(i,j)}_k}(x_s)$ be the set of leaf vertices that fan out from the vertex $x_{s}$. First we prove that $L_{T^{(i,j)}_k}(x_s)=L_{T^{(i^\prime, j^\prime)}_{k^\prime}}(x_s)$. Suppose $L_{T^{(i,j)}_k}(x_s) \neq L_{T^{(i^\prime, j^\prime)}_{k^\prime}}(x_s)$. In Condition \ref{con4} of OIC structure, we assumed that all the $I$-paths passing through a non-inner vertex must terminate at the vertices belonging to only one vertex subset $V_I^{(i,j)}$. Hence, we have $L_{T^{(i,j)}_k}(x_s)$ is a subset of $V_k^{(i,j)}\setminus \{x_{(i,j),k},x_{(i^\prime, j^\prime),k^\prime}\}$. This follows from the fact that there exists no $I$-cycle with the vertices $x_{(i,j),k}$ and $x_{(i^\prime, j^\prime),k^\prime}$ (Condition \ref{con3} of OIC definition). Let $x_a \in V_k^{(i,j)}\setminus \{x_{(i,j),k},x_{(i^\prime, j^\prime),k^\prime}\}$ such that $x_a \in L_{T^{(i,j)}_k}(x_s)$ but $x_a \notin L_{T^{(i^\prime, j^\prime)}_{k^\prime}}(x_s)$. Such vertex $x_a$ exists because we assumed $L_{T^{(i,j)}_k}(x_s) \neq L_{T^{(i^\prime, j^\prime)}_{k^\prime}}(x_s)$. 

In tree $T^{(i,j)}_k$, there exists a directed path from $x_{(i,j),k}$, which includes $x_s$ and the leaf vertex $x_a$. Let this path be $P_1$. Similarly, in tree $T^{(i^\prime, j^\prime)}_{k^\prime}$, there exists a directed path from $x_{(i^\prime, j^\prime),k^\prime}$, which does not includes $x_s$ (since $x_a \notin L_{T^{(i^\prime, j^\prime)}_{k^\prime}}(x_s)$) and ends at the leaf vertex $x_a$. Let this path be $P_2$. In $G$, we can also obtain a directed path from $x_{(i^\prime, j^\prime),k^\prime}$ which pass through $x_s$ and ends at the leaf vertex $x_a$. Let this path be $P_3$. The paths $P_2$ and $P_3$ are different, which indicates the existence of multiple paths from $x_{(i^\prime, j^\prime),k^\prime}$ to $x_a$. This is a contradiction from Condition \ref{con2} of OIC definition. Hence, we have 
\begin{align*}
L_{T^{(i,j)}_k}(x_s)=L_{T^{(i^\prime, j^\prime)}_{k^\prime}}(x_s). 
\end{align*}

Let $N^+_{T^{(i,j)}_k}(x_s)\neq N^+_{T^{(i^\prime, j^\prime)}_{k^\prime}}(x_s).$ We prove a contradiction. Let $x_b \in N_{T^{(i,j)}_{k}}(x_s)$ but $x_b \notin N_{T^{(i^\prime, j^\prime)}_{k^\prime}}(x_s)$ (such $x_b$ exists in the out-neighborhood of either $T^{(i,j)}_k$ or $T^{(i^\prime, j^\prime)}_{k^\prime}$ follows from the fact that we assumed that the out-neighborhood of $x_s$ is not same in both the trees). The vertex $x_b$ may belongs to $L_{T^{(i,j)}_k}(x_s)$ or it may not belongs to $L_{T^{(i,j)}_k}(x_s)$. If we assume that the vertex $x_b$ belongs to $L_{T^{(i,j)}_k}(x_s)$, in the above para we proved that this leads to multiple paths from $x_{k^\prime}^{(i^\prime, j^\prime)}$ to $x_b$. Hence, we assume $x_b$ does not belong to $L_{T^{(i,j)}_k}(x_s)$. Let $x_d$ be a leaf vertex such that $x_d \in L_{T^{(i,j)}_k}(x_b)$ and there exists a path from $x_s$ to $x_b$ to $x_d$ in $T^{(i,j)}_k$. We have a path from $x_{(i^\prime, j^\prime),k^\prime}$ to $x_s$ exists in $T^{(i^\prime, j^\prime)}_{k^\prime}$. Thus a path from  $x_{(i^\prime, j^\prime),k^\prime}\Rightarrow \ldots x_s\Rightarrow x_b \Rightarrow \ldots \Rightarrow x_d$ exists in $G$. Since $L_{T^{(i,j)}_k}(x_s)=L_{T^{(i^\prime, j^\prime)}_{k^\prime}}(x_s)$, we have $x_d \in L_{T^{(i^\prime, j^\prime)}_{k^\prime}}(x_s)$. In $T^{(i^\prime, j^\prime)}_{k^\prime}$, there exists a path from $x_{(i^\prime, j^\prime),k^\prime}$ to $x_d$, which includes $x_s$ followed by a vertex which is present in $N_{T^{(i^\prime, j^\prime)}_{k^\prime}}(x_s)$ and this path is different from $x_{(i^\prime, j^\prime),k^\prime}\Rightarrow \ldots x_s\Rightarrow x_b \Rightarrow \ldots \Rightarrow x_d$. Hence, multiple $I$-paths are present between $x_d$ and $x_{(i^\prime, j^\prime),k^\prime}$. This is a contradiction. Hence, we have $N_{T^{(i,j)}_k}(x_s)=N_{T^{(i^\prime, j^\prime)}_{k^\prime}}(x_s).$

Since,  \eqref{on} is true for every two trees comprising $x_s$, we have the out-neighborhood of $x_s$ is same in $T_k^{(i,j)}$ and $G$.
\end{proof}

\subsection*{Decoding of non-inner vertices}
For every $x_k \in V_{NI}$, receiver wanting $x_k$ can be decoded it from the index code symbol $y_k$ of \eqref{code11}. This follows from the fact that the receiver wanting $x_k$ knows all the messages present in $y_k$ other than $x_k$ (the out neighbourhood  of $x_k$ in $G$). 

\subsection*{Decoding of inner vertices}
From Lemma \ref{lemma1}, for any vertex $x_{s} \in V(T_k^{(i,j)})=V_k^{(i,j)}$ for $i \in [0:d]$ and $j \in [1:w_i]$ and $x_s \in V_{NI}$, the out-neighborhood of $x_s$ is same in $T_k^{(i,j)}$ and $G$. Hence, for every $x_{(i,j),k} \in V_I^{(i,j)}$ $i \in [0:d]$ and $j \in [1:w_i]$, the receiver wanting $x_{(i,j),k}$ can decode it by using the tree $T_k^{(i,j)}$ as shown below. 

In $T_k^{(i,j)}$, let $z_k^{(i,j)}$ be the XOR of index code symbols corresponding to all non-leaf vertices at depth greater than zero. In $T_k^{(i,j)}$, the message requested by a non-leaf vertex, say $x_{k^\prime}$, at a depth strictly greater than one appears exactly twice in  $z_k^{(i,j)}$ as described below.
\begin{itemize}
\item Once in the index code corresponding to $x_{k^\prime}$.
\item Once in the index code corresponding to the parent of $x_{k^\prime}$ in $T_k^{(i,j)}$.
\end{itemize}
Hence, they cancel out each other in $z_k^{(i,j)}$ and the terms remaining in $z_k^{(i,j)}$ are XOR of the following.
\begin{itemize}
\item Messages requested by all non-leaf vertices of $T_k^{(i,j)}$ at depth one.
\item Messages requested by all leaf vertices of $T_k^{(i,j)}$ at depth strictly greater than one.
\end{itemize}

In $T_k^{(i,j)}$, the XOR of $w_k^{(i,j)}$ with $z_k^{(i,j)}$ gives the following along with $x_{(i,j),k}$.
\begin{itemize}
\item The message requested by all non-leaf vertices at depth one, which are out-neighbors of $x_{(i,j),k}$.
\item The messages requested by all leaf vertices at depth one, which are also the out-neighbors of $x_{(i,j),k}$.
\end{itemize}
This follows from the fact that the messages requested by each leaf vertex at depth strictly greater than one in the tree $T_k^{(i,j)}$ is present in both $z_k^{(i,j)}$ and $w_k^{(i,j)}$ and hence cancel each other. Hence, XOR of $z_k^{(i,j)}$ and $w_k^{(i,j)}$ gives  $x_{(i,j),k}$ and all the vertices in the out neighbourhood of $x_{(i,j),k}$ in $T_k^{(i,j)}$. As the receiver wanting  $x_{(i,j),k}$ knows all its out neighbors as side-information, the receiver decodes $x_{(i,j),k}$ from $z_k^{(i,j)}\oplus w_k^{(i,j)}$.
\end{appendix}
\end{document}